\documentclass[11pt]{article}

\usepackage{fullpage, authblk, dsfont, graphicx, amsthm, amsfonts, amsmath, amssymb}
\usepackage[numbers,sort&compress]{natbib}
\usepackage[font=small,labelfont=bf]{caption}

\usepackage{hyperref}
\newcommand{\expref}[2]{\texorpdfstring{\hyperref[#2]{#1~\ref{#2}}}{#1~\ref{#2}}}

\usepackage{color}
\definecolor{webgreen}{rgb}{0,.5,0}
\definecolor{webblue}{rgb}{0,0,.5}

\input{qcircuit.sty}
\newcommand{\id}{\mathds{1}}
\newcommand{\tr}{\mathrm{Tr}}
\newcommand{\bra}[1]{\langle #1|}
\newcommand{\ket}[1]{|#1\rangle}

\newcommand{\eq}[1]{(\ref{#1})}

\newcommand{\trclose}{\mathrm{tr}}
\def\End{\mathrm{End}}
\def\GL{\mathrm{GL}}

\def\CC{\mathbb{C}}
\def\ZZ{\mathbb{Z}}

\newcommand{\innerprod}[2]{\langle #1| #2\rangle}

\newtheorem{theorem}{Theorem}
\newtheorem{fact}{Fact}
\newtheorem{lemma}{Lemma}
\newtheorem{definition}{Definition}

\begin{document}

\title{Yang-Baxter operators need quantum entanglement\\ to distinguish knots}
\author[1]{Gorjan Alagic}
\author[2,4]{Michael Jarret}
\author[3,4]{Stephen P. Jordan}
\affil[1]{\small{Department of Mathematical Sciences, University of Copenhagen, Copenhagen, DK}}
\affil[2]{\small{Department of Physics, University of Maryland, College Park, MD, USA}}
\affil[3]{\small{Applied and Computational Mathematics Division, National Institute of Standards and Technology, Gaithersburg, MD, USA}}
\affil[4]{\small{Joint Center for Quantum Information and Computer Science (QuICS), University of Maryland, College Park, MD, USA}}

\date{}

\bibliographystyle{unsrt}

\maketitle

\begin{abstract}
Any solution to the Yang-Baxter equation yields a family of
representations of braid groups. Under certain conditions, identified
by Turaev~\cite{T88}, the appropriately normalized trace of these
representations yields a link invariant. Any Yang-Baxter solution can
be interpreted as a two-qudit quantum gate. Here we show that if this
gate is non-entangling, then the resulting invariant of knots is
trivial. We thus obtain a general connection between topological
entanglement and quantum entanglement, as suggested by Kauffman
\emph{et al.}~\cite{K02}.
\end{abstract}

%-------------------------------------------------------------
%-------------------------------------------------------------
\section{Introduction}
%-------------------------------------------------------------
%-------------------------------------------------------------

Based on a variety of suggestive phenomena~\cite{A97, K02, KL02, KL02b, KL03, KL04, KL04b, K05, ZKG05}, Kauffman \emph{et al.} have proposed that there may be deep connections between quantum entanglement and topological entanglement. In this work, we obtain a rather general theorem demonstrating one such connection. 

The basis for this connection is a procedure introduced by Turaev, whereby one obtains a link invariant from any invertible operator which satisfies the Yang-Baxter equation~\cite{T88}. Roughly speaking, one first reorganizes the link into the closure of a braid, which can then be interpreted as the spacetime trajectory of a fixed number of particles living in a two-dimensional medium. Each braiding of a pair of adjacent particles is then replaced by a copy of the Yang-Baxter operator; finally, one computes a certain normalized trace of the operator describing the total evolution. By selecting particular solutions of the Yang-Baxter equation, this procedure yields as special cases all quantum invariants of links; some of these invariants, in turn, have played a crucial role in the study of topological quantum field theories (TQFTs) as initiated by Witten~\cite{Witten89}. In particular, one recovers both the Tsohantjis-Gould invariants~\cite{TG94, KR12} and the HOMFLY polynomial~\cite{HOMFLY}; the HOMFLY polynomial, in turn, yields the Jones and Alexander polynomials as special cases~\cite{K91}. Thus, the procedure of Turaev encompasses many of the best-known link invariants, both classical and modern. 

The relevant solutions to the Yang-Baxter equation are invertible operators on $V \otimes V$ for some $d$-dimensional vector space $V$. If the base field is the complex numbers, we may interpret such matrices as quantum gates\footnote{Physically, quantum gates are unitary; however, the notion of generating entanglement extends in an obvious way to invertible operators. Moreover, restricting to unitary solutions yields only special cases of some important invariants. For example, unitary solutions are only known to yield Jones polynomials evaluated at roots of unity.} acting on two $d$-level systems. Here, we prove that if this operator is non-entangling, then the link invariant resulting from Turaev's procedure takes the same value on all knots. Hence, in this context, one may say that quantum entanglement is necessary to detect topological entanglement.

%-------------------------------------------------------------
\paragraph{Yang-Baxter operators and braids.}
The Yang-Baxter equation and its variants arise naturally in many areas of physics~\cite{PA06}. Here, we consider the constant Yang-Baxter equation, which is defined as follows. Let $V$ be a finite-dimensional complex Hilbert space, and let $R$ be a linear operator on $V \otimes V$. Then $R$ satisfies the (constant) Yang-Baxter equation if
\begin{equation}
\label{yb}
(R \otimes \id)(\id \otimes R)(R \otimes \id) = 
(\id \otimes R)(R \otimes \id)(\id \otimes R)
\end{equation}
where $\id$ is the identity operator on $V$. A linear operator $R$ satisfying \eq{yb} is called a \emph{Yang-Baxter operator}.

Recall that in the Artin presentation, the braid group $B_n$ is generated by the elementary adjacent crossings $\sigma_1,\ldots,\sigma_{n-1}$ of the $n$ strands, subject to the relations
\begin{eqnarray}
\sigma_i \sigma_j & = & \sigma_j \sigma_i  \quad \quad \quad \ \ |i-j| \geq 2
\label{R1} \\
\sigma_i \sigma_{i+1} \sigma_i & = & \sigma_{i+1} \sigma_i \sigma_{i+1}
\quad i=1,\ldots,n-2 \label{R2}.
\end{eqnarray}
An invertible Yang-Baxter operator $R$ yields an infinite family of braid group representations  
\begin{eqnarray}\label{rep1}
\rho_n^{(R)} &: B_n & \to   ~~GL(V^{\otimes n})\\
\label{rep2}
\rho_n^{(R)} &: \sigma_j & \mapsto   ~~\id_V^{\otimes j-1} \otimes R
\otimes \id_V^{\otimes n-j-1}\,,
\end{eqnarray}
one for every $n$. As an abbreviation, we let $R_i = \rho_n^{(R)}(\sigma_i)$. One easily verifies that $\rho_n^{(R)}$ is a representation of $B_n$ by noting that $R_i R_j = R_j R_i$ for $|i-j| \geq 2$ due to \eq{rep2} and $R_i R_{i+1} R_i = R_{i+1} R_i R_{i+1}$ for $n=1,\ldots,n-2$ due to \eq{yb}. 

%-------------------------------------------------------------
\paragraph{Turaev's construction.}

Alexander's theorem states that any oriented link is equivalent (in the standard sense, i.e., ambient isotopic) to the trace closure $b^\trclose$ of some braid $b$~\cite{A23}. The trace closure is illustrated in \expref{Figure}{fig:trace-closure}. Markov's theorem gives necessary and sufficient conditions for a function on braids to define an invariant of oriented links via the trace closure.

\begin{figure}[h]
\begin{center}
\includegraphics{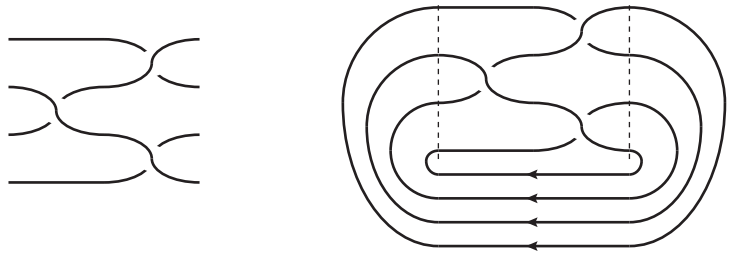}
\caption{\label{fig:trace-closure} Alexander's theorem states that any oriented link can be presented as the trace closure (right) of some braid (left). }
\end{center}
\end{figure}

\begin{theorem}[\cite{M35}]
\label{markov_thm}
Let $f$ be a function on braids. Define a function $g$ on oriented links by $g(L) = f(b)$ where $b$ is any braid such that $b^\trclose$ is equivalent to $L$. Then $g$ is well-defined as a function on equivalence classes (i.e., is an invariant of oriented links) if and only if $f$ is invariant under the following so-called Markov moves:
\begin{eqnarray}
b & \to & a b a^{-1} \textrm{ \rm{for} $a \in B_n$} \label{move1}\\
b & \to & i_{n+1}(b) \sigma_n \label{move2a} \\
b & \to & i_{n+1}(b) \sigma_n^{-1} \label{move2b}
\end{eqnarray}
where $i_{n+1} : \sigma_j \mapsto \sigma_j$ denotes the canonical inclusion of $B_n$ into $B_{n+1}$.
\end{theorem}

If $\{\rho_n:B_n \rightarrow \GL(V_n) ~|~ n=1,2,\ldots\}$ is a family of braid group representations, then one sees that the function $f(b)=\tr[\rho_n(b)]$ is invariant under the Markov move \eq{move1} due to the cyclic property of the trace. In some cases, by including sufficient correction factors, functions of the form $\tr[\rho_n(b)]$ can be made invariant under all of the Markov moves, thus yielding link invariants.

Specifically, given $R \in \GL(V \otimes V)$ satisfying the constant Yang-Baxter equation, Turaev defines the following (parameterized) function on braids \cite{T88}.
\begin{equation}\label{eq:pre-turaev-invariant}
I_{\mathcal{R}}(b) = \alpha^{-w(b)} \beta^{-n} \tr \left[ \rho_n^{(R)}(b) \cdot \mu^{\otimes n} \right]\,,
\end{equation}
where $\alpha,\beta$ are invertible scalars, $\mu \in \End(V)$, and $w : B^n \rightarrow \ZZ$ is the homomorphism defined by setting $w :\sigma_j^{\pm} \mapsto \pm 1.$ In this context, the function $w$ coincides with the notion of the writhe of a link diagram; this is simply the number of crossings of the form \includegraphics[width=0.15in]{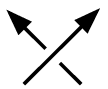} minus the number of crossing of the form \includegraphics[width=0.15in]{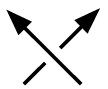}. Collectively, the set of parameters $\mathcal{R} = (R,\alpha,\beta,\mu)$ defining \eqref{eq:turaev-invariant} is called an \emph{enhanced Yang-Baxter operator}, provided the following three conditions are satisfied.
\begin{eqnarray}
\label{enh1}
\mu \otimes \mu \cdot R & = & R \cdot \mu \otimes \mu \\
\label{enh2}
{\tr}_2(R \cdot \mu \otimes \mu) & = & \alpha \beta \mu \\
\label{enh3}
{\tr}_2(R^{-1} \cdot \mu \otimes \mu) & = & \alpha^{-1} \beta \mu,
\end{eqnarray}
Here, ${\tr}_2$ denotes the partial trace over the second tensor factor. 

As Turaev notes ~\cite[Remark 3.3]{T88}, an immediate simplification of the above is possible: one easily checks that if $\mathcal R = (R, \alpha, \beta, \mu)$ is an enhanced Yang-Baxter operator, then $\mathcal R' = (\alpha^{-1}R, 1, 1, \beta^{-1}\mu)$ is also an enhanced Yang-Baxter operator, with $I_\mathcal R = I_{\mathcal R'}$. This leads to the following simplified definition.

\begin{definition}\label{def:eyb}
Let $V$ be a finite-dimensional complex Hilbert space, $R \in \GL(V \otimes V)$ a Yang-Baxter operator, and $\mu \in \End(V)$. If $R$ commutes with $\mu \otimes \mu$ and
$$
\tr_2(R \cdot \mu \otimes \mu)  = \tr_2(R^{-1} \cdot \mu \otimes \mu) = \mu,
$$
then we say that the pair $\mathcal R = (R, \mu)$ is an enhanced Yang-Baxter operator. In that case, given any braid $b$ we define
\begin{equation}
\label{eq:turaev-invariant}
I_{\mathcal R}(b) = \tr \left[ \rho_n^{(R)}(b) \cdot \mu^{\otimes n} \right]\,.
\end{equation}
\end{definition}

\noindent As Turaev shows, each enhanced Yang-Baxter operator yields a link invariant via \eqref{eq:turaev-invariant}; the proof simply verifies the conditions of \expref{Theorem}{markov_thm} via straightforward calculations.

\begin{theorem}{\cite{T88}}
\label{turaev_thm}
If $\mathcal{R}$ is an enhanced Yang-Baxter operator then $I_{\mathcal{R}}(b^{\mathrm{tr}})$ is an invariant of oriented links.
\end{theorem}

%-------------------------------------------------------------
\paragraph{The present work.}

To state our result, we recall some basic terminology from quantum information. We say that a vector $w \in V \otimes V$ is a \emph{product state} if $w = a \otimes b$ for some $a, b \in V$; a vector which is not a product state is said to be \emph{entangled}. We also say that an invertible linear operator on $V \otimes V$ is \emph{non-entangling} if it maps product states to product states. For normalized states and unitary operators, this coincides with the standard definitions regarding entanglement in quantum theory. Using this terminology, our result can be stated as follows.

\begin{theorem}
\label{mainthm}
Let $\mathcal{R} = (R, \mu)$ be an enhanced Yang-Baxter operator with invertible $\mu$. If $R$ is non-entangling, then $I_{\mathcal{R}}$ is constant on knots.
\end{theorem}

\noindent Note that, as shown in \expref{Appendix}{noninvert}, invariants obtained from non-invertible $\mu$ can always be re-expressed using invertible $\mu$. Thus, in restricting to invertible $\mu$ in \expref{Theorem}{mainthm} there is no loss of essential generality.

In addition to the above, we show that for links, the invariant arising from any non-entangling Yang-Baxter operator admits a straightforward and efficient formula: it is a product over the components of L of the trace of a certain operator on $V$. This operator is easy to read off from the trace closure of the circuit diagram, as described in \expref{Lemma}{lem:swap-product} below. This allows us to calculate the invariant exactly in time polynomial in the number of strands, the length of the braid, and $\dim V$.

Before continuing to the proofs, we make a brief remark regarding circuit diagram notation, which we will make use of throughout the paper. It is often convenient to think about the operators $\rho_n^{(R)}(b)$ (and other operators on tensor powers of $V$) in terms of their circuit diagrams. To draw such a diagram, start with a braid diagram for $b$, and replace each positive crossing with a box labeled ``$R$'' and each negative crossing with a box labeled ``$R^{-1}$.'' When we know more about $R$ (e.g., it has the form $F \otimes G$), we can expand the diagram by replacing each of these boxes by a circuit diagram for $R$ or $R^{-1}$, as appropriate. Following conventions, we will always read circuit diagrams left-to-right; note that this is the opposite of how we write the corresponding composition of linear operators. The same conventions will apply for braid diagrams (left-to-right) and words in the braid generators (right-to-left).

%-------------------------------------------------------------
%-------------------------------------------------------------
\section{Proofs}
%-------------------------------------------------------------
%-------------------------------------------------------------

%-------------------------------------------------------------
\subsection{Characterization of non-entangling operators}
%-------------------------------------------------------------

We first describe a simple and well-known\footnote{When restricted to
  unitary gates, this characterization is sometimes referred to as
  Brylinski's theorem \cite{BB02}.} characterization of
non-entangling operators. We will need the following result of Marcus
and Moyls~\cite{marcus1959}.

\begin{theorem}\label{marcus-moyls-thm}
Let $T$ be a linear transformation on the space $M_n$ of $n \times n$ complex matrices. If the set of rank one matrices is invariant under $T$, then there exist invertible matrices $A$ and $B$ in $M_n$ such that either $T(X) = AXB$ for all $X \in M_n$, or $T(X) = AX^{\intercal} B$ for all $X \in M_n$. 
\end{theorem}

Recall that non-entangling operators are precisely the invertible elements of $\End(V \otimes V)$ which map product states to product states. We can characterize all non-entangling operators as follows. Informally speaking, the characterization states that the only nontrivial example is the swap gate $S : a \otimes b \mapsto b \otimes a$.

\begin{theorem}
Let $V$ be a finite-dimensional complex vector space, and $M \in \GL(V \otimes V)$ a non-entangling operator. Then there exist $A, B \in \GL(V)$ such that either $M = A \otimes B$ or $M = (A \otimes B) \circ S$.
\end{theorem}
\begin{proof}
We will make use of the isomorphism $\phi: V \otimes V \rightarrow \End(V)$ defined by $\phi(a \otimes b) = a \otimes b^*$ (where $b^*$ denotes the dual of $b$), as well as the induced isomorphism $\Phi : \End(V \otimes V) \rightarrow \End(\End(V))$ on operators. Using the fact that $\End(V \otimes V) \cong \End(V) \otimes \End(V)$, we may explicitly write the latter isomorphism as $\Phi(A \otimes B) : X \mapsto AXB$ for all $X \in \End(V)$.

Note that $\phi$ maps product states to rank-one matrices. It follows that $\Phi$ maps non-entangling operators on $V \otimes V$ to operators on $\End(V)$ which preserve rank one matrices. By \expref{Theorem}{marcus-moyls-thm}, there exist invertible $A, B \in \End(V)$ such that either (i.) $\Phi(M):X \mapsto AXB$ or (ii.) $\Phi(M): X \mapsto AX^{\intercal}B$.  By the definition of $\Phi$, we see that $M = A \otimes B$ in the first case, and $M = (A \otimes B) \circ S$ in the second case.
\end{proof}

In terms of circuit diagrams, non-entangling operators $R$ are expressible in one of two forms:
\[
\includegraphics[width=2in]{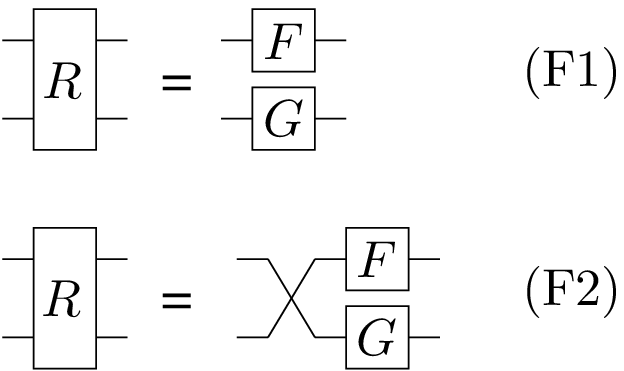}
\]
where $\includegraphics[width=0.17in]{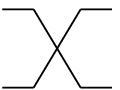}$ denotes the swap gate $S$. We will also make frequent use of the following straightforward fact about product states; it follows directly from the definition $\innerprod{a \otimes b}{c \otimes d} = \innerprod{a}{c}\innerprod{b}{d}$ of the inner product on tensor products of Hilbert spaces.

\begin{fact}\label{fact:equal-product-states}
Let $V$ be a finite-dimensional complex Hilbert space, and $a, b, c, d \in V$ such that $a \otimes b = c \otimes d$. Then $a = \alpha c$ and $b = \alpha^{-1} d$ for some $\alpha \in \CC$.
\end{fact}

%-------------------------------------------------------------
\subsection{The product case}
%-------------------------------------------------------------

Returning to the proof of the main theorem, we first handle the simpler case of product operators, that is, operators of the form (F1).

\begin{lemma}\label{lem:product-case}
Let $\mathcal R = (R, \mu)$ be an enhanced Yang-Baxter operator with $R = F \otimes G$. Then $I_{\mathcal R}$ is constant on knots.
\end{lemma}
\begin{proof}
Substituting $R = F \otimes G$ into the Yang-Baxter equation \eqref{yb} yields
\begin{equation}
\Qcircuit @C=1em @R=.5em {
&\gate{F} &\qw	&\gate{F} &\qw\\
&\gate{G} &\gate{F} &\gate{G} &\qw\\
&\qw &\gate{G} &\qw &\qw\\
}
\raisebox{-17pt}{\qquad = \qquad}
\Qcircuit @C=1em @R=.5em {
&\qw &\gate{F}	&\qw &\qw\\
&\gate{F} &\gate{G} &\gate{F} &\qw\\
&\gate{G} &\qw &\gate{G} &\qw\\
}
\end{equation}
By applying \expref{Fact}{fact:equal-product-states} twice to the above,  we see that $F^2 = t F$ and $G^2 = s G$ for some constants $t, s$. By invertibility of $F$ and $G$, $F = t \id_V$ and $G = s \id_V$ and hence $R = r \id_{V \otimes V}$ where $r := ts$. Thus $I_{\mathcal R}(b) = r^{w(b)} \tr[\mu]^n$. 

If $\tr[\mu] = 0$ we are done, so assume $\tr[\mu] \neq 0$. Note that $\sigma_1^\trclose$ and $(\sigma_1^{-1})^\trclose$ (as trace-closures of elements of $B_2$) are both equivalent to the unknot $K$. Then $r^{-1} \tr[\mu]^2 = I_\mathcal R(K) = r \tr[\mu]^2$, and hence $r = \pm 1$. Finally, choose any knot $L$ and write it as the trace closure $b^\trclose$ of some braid $b$. As shown in \cite{K91}, $b^\trclose$ can be deformed into the unknot $K$ by changing some of the undercrossings of $b$ into overcrossings (or vice-versa.) However, each such replacement does not affect $w(b) \bmod 2$ or $n$; it follows that $I_\mathcal R(L) = I_\mathcal R(K)$. 
\end{proof}

%-------------------------------------------------------------
\subsection{The swap case}
%-------------------------------------------------------------

We now assume that $R$ has the form (F2), i.e., $R = (F \otimes G) \circ S$ for some $F, G \in \GL(V)$. We will first show that establishing this case amounts to understanding all of the commutation relations between $F$, $G$, $\mu$, and their inverses. Recall that the operator $\rho_n^{(R)}(b) \cdot \mu^{\otimes n} \in \End(V^{\otimes n})$ is associated with a corresponding circuit diagram, constructed in the obvious way from the braid diagram for $b$. One can also draw a trace closure for this circuit diagram by simply connecting the wires of the diagram just as we connect the strands of a braid; an example is shown in \expref{Figure}{fig:circuit-closure}. This diagrammatic picture can lead to some nice insights, as in the following lemma.

\begin{figure}[h]
\begin{center}
\includegraphics{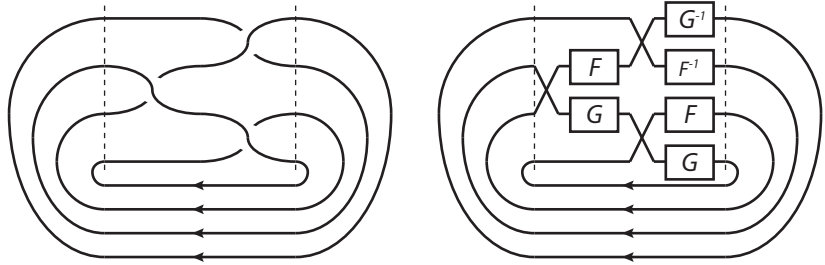}
\caption{\label{fig:circuit-closure} The trace closure of a braid (left) and the corresponding ``trace-closed'' circuit diagram (right), for a Yang-Baxter operator of the form (F2), with $\mu = \id$. By \expref{Lemma}{lem:swap-product}, the invariant is $\text{Tr}\bigl[ F^{-1}GGFFG^{-1} \bigr]$, where the sequence of operators is found by following the (in this case single) oriented wire.}
\end{center}
\end{figure}

\begin{lemma}\label{lem:swap-product}
Let $\mathcal R = (R, \mu)$ be an enhanced Yang-Baxter operator with $R = (F \otimes G) \circ S$. Let $b \in B^n$ and let $m$ be the number of components of $b^\trclose$. Then
$$
I_\mathcal R (b^\text{tr}) = \prod_{j=1}^m \emph{Tr}\bigl[ A_j \bigr]\,,
$$
where $A_j \in \End(V)$ is the product of elements of $M := \{F, G, F^{-1}, G^{-1}, \mu\}$ found by ``following the wire'' along the $j$th component of the trace-closed circuit diagram for $\rho_n^{(R)}(b) \cdot \mu^{\otimes n}$.
\end{lemma}
This lemma is quite intuitive if one regards diagrams of the sort
illustrated in figure \ref{fig:circuit-closure} as tensor networks to
be contracted. Nevertheless, we give a formal linear-algebraic proof
as follows.
\begin{proof}
We first write
$$
\tr\bigl[\rho_n^{(R)}(b) \cdot \mu^{\otimes n}\bigr]
= \sum \bra{\lambda_1 \otimes \cdots \otimes \lambda_n} \rho_n^{(R)}(b) \cdot \mu^{\otimes n} \ket{\lambda_1 \otimes \cdots \otimes \lambda_n}\,,
$$
where each $\lambda_j$ ranges over a fixed basis of $V$. We then expand $\rho_n^{(R)}(b)$ into images of braid generators, and apply the resulting sequence of operators to the vector $\ket{\lambda_1 \otimes \cdots \otimes \lambda_n}$. The braid generators are mapped to $R$, while inverse generators are mapped to
\begin{equation}
\label{rinv}
R^{-1} = \begin{array}{c} \includegraphics[width=0.7in]{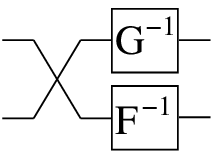} \end{array}.
\end{equation}
So, for each generator appearing in $b$, we exchange two tensor product components, and then apply a pair of elements of $M$ to those components: $F$ and $G$ for positive crossings, and $G^{-1}$ and $F^{-1}$ for negative crossings. The final result is an expression of the form
\begin{align}\label{eq:trace-closed-product}
\tr\bigl[\rho_n^{(R)}(b) \cdot \mu^{\otimes n}\bigr]
&= \sum \innerprod{\lambda_1 \otimes \cdots \otimes \lambda_n}{M_1 \lambda_{\pi(1)} \otimes \cdots \otimes M_n \lambda_{\pi(n)}}\notag \\
&= \sum \bra{\lambda_1}M_1\ket{\lambda_{\pi(1)}} \cdots \bra{\lambda_n}M_n\ket{\lambda_{\pi(n)}}\,,
\end{align}
where each $M_j$ is a product of elements of $M$, and $\pi \in S_n$.

The map $\pi$ is the permutation on strands induced by the braid $b$. Consider the braid diagram of $b$, drawn left-to-right. The strand beginning in position $1$ on the left terminates in position $\pi(1)$ on the right. If we take the trace closure of the diagram and continue following this strand, we will arrive at position $\pi(1)$ on the left, then $\pi^2(1)$ on the right, and so on until we finally return to position $1$ on the left. This completes a component of $b$ and a cycle of the permutation $\pi$. We will organize the sum \eqref{eq:trace-closed-product} by grouping all of the terms that appear in each component (and cycle) together. Note that these groupings involve disjoint sets of the $\lambda_k$, as the indices $k$ are partitioned by the cycles of $\pi$. We can thus separate the sums involving these disjoint sets, leaving us with a product
$$
\tr\bigl[\rho_n^{(R)}(b) \cdot \mu^{\otimes n}\bigr] = \prod_{j=1}^m a_j
$$
over the components of $b^\text{tr}$. Each term $a_j$ is straightforward to understand separately. For simplicity, let us only consider the component $j=1$ that contains the strand which begins in position $1$ on the left. We rearrange the terms involved so that they occur in the order determined by $\pi$. We then have
\begin{align*}
a_1 
&= \sum \bra{\lambda_1}M_1\ket{\lambda_{\pi(1)}} \bra{\lambda_{\pi(1)}}M_{\pi(1)}\ket{\lambda_{\pi^2(1)}} \cdots \bra{\lambda_{\pi^{l-1}(1)}}M_{\pi^{l-1}(1)}\ket{\lambda_{1}}\\
&= \tr\bigl[ M_1 M_{\pi(1)} \cdots M_{\pi^{l-1}(1)}\bigr]\,,
\end{align*}
where $l$ is the size of the cycle of $\pi$ containing $1$, i.e. $\pi^l(1) = 1$. From this expression, it is clear how to compute the operators $A_j$ in the theorem statement. First draw the circuit diagram for $\rho_n^{(R)}(b)$, and connect the wires together just as in the trace closure of a braid. Start at any point along any wire in the $j$th component, and follow that wire until back to the same starting point. As this is done, write down each operator from $M$ encountered along the way, in order. Their product is precisely the operator $A_j$.
\end{proof}

We remark that, in the case where $b^\trclose$ is a knot, the
permutation $\pi$ in the proof is an $n$-cycle, and the entire
invariant becomes simply the trace of the product of operators from
$M$ found by ``following the wire'' as above, just as in the example
in \expref{Figure}{fig:circuit-closure}. Thus, for a knot, the effect
of changing an undercrossing of $b$ into an overcrossing is to replace
a pair of operators $F$ and $G$ in the product being traced over with
$G^{-1}$ and $F^{-1}$, respectively. We will use this fact, and the
relations among the elements of $M$, to show that such a change has no
effect on the value of the invariant.  

\begin{lemma}\label{lem:commutation-relations}
Let $(R, \mu)$ be an enhanced Yang-Baxter operator, with $R$ of the form (F2) and $\mu$ invertible. Then the operators $F, G, \mu$ all pairwise commute, and $GF = (GF)^{-1}$.
\end{lemma}

\begin{proof}

We first show that $F$ and $G$ must commute. Substituting (F2) into the left-hand side of the Yang-Baxter equation and evaluating on an arbitrary product state $\ket{a} \ket{b} \ket{c}$ yields:
\[
\begin{array}{c}\includegraphics[width=2.7in]{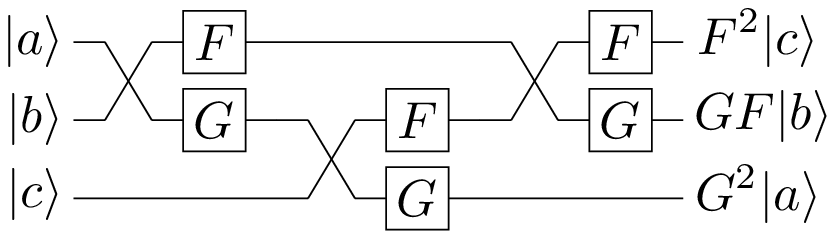} \end{array} .
\]
The right-hand side yields
\[
\begin{array}{c} \includegraphics[width=2.7in]{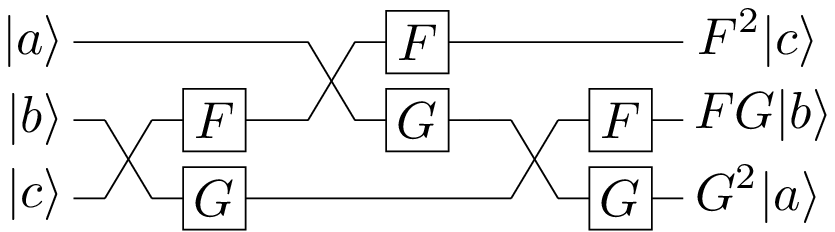} \end{array} .
\]
Thus $GF \ket{b} = FG \ket{b}$ for all $\ket{b}$, which implies $GF = FG$.

Next we consider the commutation relations of $\mu$ with $F$ and $G$. Recall that $R$ must commute with 
$\mu \otimes \mu$, so that
\begin{equation}
(G \mu) \otimes (F \mu) = (\mu G) \otimes (\mu F).
\end{equation}
Hence there exists invertible $q \in \CC$ such that
\begin{eqnarray}
F \mu & = & q \mu F \label{fcom} \\
G \mu & = & q^{-1} \mu G \label{gcom}\,.
\end{eqnarray}
\noindent We will show that  $q = 1$. Recall that $R$ must satisfy the partial trace condition \eq{enh2}, i.e.,
\begin{equation}
\sum_{k=1}^d \begin{array}{c} \includegraphics[width=1.4in]{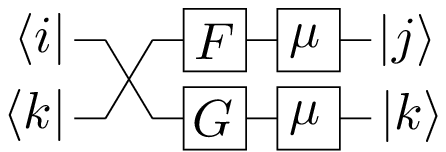} \end{array}
= \bra{j} \mu \ket{i}.
\end{equation}
That is,
\begin{equation}
\sum_{k=1}^d \bra{j} \mu F \ket{k} \bra{k} \mu G \ket{i} = \bra{j} \mu \ket{i}.
\end{equation}
Using the completeness relation $\sum_{k=1}^d \ket{k} \bra{k} = \id$, we see that
\begin{equation}\label{gmfm}
\mu F \mu G = \mu.
\end{equation}
Simplifying the above using only the invertibility of $\mu$, $F$, and $G$, we have
\begin{equation}\label{gmfmA}
\mu^{-1} = GF\,.
\end{equation}
If instead we apply \eqref{fcom} to \eqref{gmfm}, we get $\mu^{-1} = q FG$. Since $F$ and $G$ commute, we conclude that $q = 1$.

Finally, consider the second partial trace condition \eq{enh3}. Recalling the form \eq{rinv} of $R^{-1}$, we have
\begin{equation}
\sum_{k=1}^d \begin{array}{c} \includegraphics[width=1.4in]{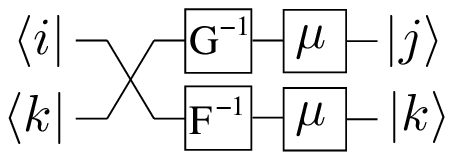} \end{array} = \bra{j} \mu \ket{i}
\end{equation}
which implies
\begin{equation}
\label{fimugimu}
\mu G^{-1} \mu F^{-1} = \mu.
\end{equation}
As before, we simplify using only the invertibility of $\mu$, $F$, and $G$, to get
\begin{equation}\label{gmfmB}
\mu^{-1} = (GF)^{-1}\,.
\end{equation}
Combining \eqref{gmfmB} with \eqref{gmfmA} proves the last claim in
the lemma statement.
\end{proof}

We are now prepared to prove the main theorem, for the case of form (F2).

\begin{lemma}
Let $\mathcal R = (R, \mu)$ be an enhanced Yang-Baxter operator with
$R = (F \otimes G) \circ S$ and invertible $\mu$. Then $I_\mathcal R$ is constant on knots.
\end{lemma}
\begin{proof}
Let $L$ be a knot and $b$ a braid such that $b^\trclose$ is equivalent to $L$. As shown in~\cite{K91}, $L$ can be disentangled into the unknot by
replacing some of the undercrossings in $b$ with overcrossings (or vice-versa.) Each such replacement in the diagram corresponds to applying the replacement rule $\sigma_j^\pm \mapsto \sigma_j^\mp$ to some word for $b$. In the representation $\rho_n^{(R)}$, the corresponding change is $R_j^{\pm} \mapsto R_j^\mp.$ We will show that transformations of this kind leave the invariant $I_\mathcal R$ unchanged.

By \expref{Lemma}{lem:swap-product}, the invariant evaluated at $L$ takes the form
$$
I_\mathcal R (b^\text{tr}) = \tr \bigl[ A_1 A_2 \cdots A_l \bigr]\,,
$$
where each $A_j$ is an element of $M = \{F, F^{-1}, G, G^{-1}, \mu \}$. By \expref{Lemma}{lem:commutation-relations}, all of the operators in the sequence $\{A_j\}_j$ commute. Let $b_0$ be the braid produced by changing some particular undercrossing of $b$ into an overcrossing. As we saw in the proof of \expref{Lemma}{lem:swap-product}, this change has the following effect on the sequence $\{A_j\}_j$: for some pair of indices $(t, s)$ determined by the position of the undercrossing, we replace $A_t = F$ with $G^{-1}$, and we replace $A_s = G$ with $F^{-1}$. By commutativity of the $A_j$, we then have
$$ 
I_\mathcal R (b^\text{tr}) = \tr \bigl[ FG A \bigr]
\qquad \text{and} \qquad
I_\mathcal R (b_0^\text{tr}) = \tr \bigl[ G^{-1}F^{-1} A \bigr]\,,
$$
where
$$
A = \prod_{j \notin \{s, t\}} A_j\,.
$$
By \expref{Lemma}{lem:commutation-relations}, $FG =
G^{-1}F^{-1}$. Hence $I_\mathcal R (b^\text{tr}) = I_\mathcal R
(b_0^\text{tr})$. Thus we have shown that $I_\mathcal R$ is unchanged
when overcrossings are changed to undercrossings. Similarly,
$I_\mathcal R$ is unchanged when undercrossings are changed to
overcrossings. It follows that $I_\mathcal R$ is constant on knots.
\end{proof}

\section{Concluding Remarks}

In the proof of \expref{Theorem}{mainthm} we saw that if $R$ is of the
form (F1) then $I_{\mathcal{R}}(L)$ can depend only on the number of
components in the link $L$. In the case that $R$ is of the form (F2)
we obtained only the weaker statement that $I_{\mathcal{R}}(K)$ takes
the same value on all knots (\emph{i.e.} single-component links)
$K$. It is tempting to conjecture that even in case (F2),
$I_{\mathcal{R}}(L)$ can depend only on the number of components of
$L$. However, this is not true in general, as shown by the following
counterexample. Let $\mathcal{R} = (R,1,1,\mu)$ where $R$ is of the
form (F2) and 
\begin{eqnarray}
F & = & \left[ \begin{array}{cc} 1 & 0 \\ 0 & 1 \end{array} \right] \\
G & = & \left[ \begin{array}{cc} 1 & 0 \\ 0 & -1 \end{array} \right] \\
\mu & = & G.
\end{eqnarray}
One can verify that $\mathcal{R}$ is an enhanced Yang-Baxter operator
and $I_{\mathcal{R}}$ evaluated on the Hopf link yields 4, whereas
$I_{\mathcal{R}}$ evaluated on the two-component unlink yields 0.

Our main result (\expref{Theorem}{mainthm}) shows that to obtain a
nontrivial knot invariant by Turaev's construction it is necessary for
the Yang-Baxter operator to be entangling. It is thus natural to ask
whether this is also a sufficient condition. We can see that this
is not the case by constructing an example of an entangling
Yang-Baxter solution $R_e$ and a corresponding enhanced Yang-Baxter
operator $\mathcal{R}_e$ such that $I_{\mathcal{R}_e}$ is
trivial. Specifically, with the aid of \cite{H93, D03}, we find
the following example. Let
\begin{eqnarray}
R_e & = & \left[ \begin{array}{cccc}
1 & 0 & 0 & 0 \\
0 & 0 & 1 & 0 \\
0 & 1 & 0 & 0 \\
0 & 0 & 0 & -1
\end{array} \right] \\
\mu_e & = & \left[ \begin{array}{cc}
1 & 0 \\
0 & -1 \end{array}
\right],
\end{eqnarray}
with $\alpha = \beta = 1$. Interpreted as a quantum gate, $R_e$ is
entangling. (In quantum computation language it is a SWAP gate
followed by a controlled-phase gate.) One can verify that $R_e$
satisfies the Yang-Baxter equation and $\mathcal{R}_e$ is an enhanced
Yang-Baxter operator (\emph{i.e.} satisfies \eq{enh1}, \eq{enh2}, and
\eq{enh3}).  Thus, by \expref{Theorem}{turaev_thm},
$I_{\mathcal{R}_e}$ is a link invariant. However, $R_e =
R_e^{-1}$. Thus, $I_{\mathcal{R}_e}$ does not distinguish
overcrossings from undercrossings. Therefore $I_{\mathcal{R}_e}(L)$
depends only on the number of components of link $L$.

\vspace{11pt}
\noindent
\textbf{Acknowledgments:} Portions of this paper are a contribution of NIST, an agency of the US government, and are not subject to US copyright.

\bibliography{ybtrivial}

\appendix

\section{Non-invertible $\mu$}
\label{noninvert}

In this appendix, we show that the case of non-invertible $\mu$
reduces to the case of invertible $\mu$. Thus, by restricting to
invertible $\mu$, as we have done in our main theorem, there is no
loss of essential generality.

\begin{lemma}
\label{invertible}
Suppose $R$ is a Yang-Baxter solution and $\mathcal{R} =
(R,\alpha,\beta,\mu)$ is an enhanced Yang-Baxter operator yielding
link invariant $I_{\mathcal{R}}$. Suppose $\mu$ is not invertible but
$I_{\mathcal{R}}$ is not identically zero. Then there exists some
lower-dimensional enhanced Yang-Baxter operator $\mathcal{R}' =
(R',\alpha,\beta,\mu')$ with invertible $\mu'$ such that
$I_{\mathcal{R}'} = I_{\mathcal{R}}$.
\end{lemma}

\begin{proof}
Let $V$ denote the range of the linear operator $\mu$. By \eq{enh1},
the range of $R \cdot \mu \otimes \mu$ is the same as the range of
$\mu \otimes \mu \cdot R$, namely $V \otimes V$. Thus, $R$ maps $V
\otimes V$ to itself, as does $R^{-1}$. We can therefore replace each
instance of $R$ in the trace defining $I_{\mathcal{R}}$ with its
restriction $R'$ to $V \otimes V$, and similarly replace $R^{-1}$ and
$\mu$. The invariant $I_{\mathcal{R}'}$ with $\mathcal{R}' =
(R',\alpha,\beta,\mu')$ will be identically equal to
$I_{\mathcal{R}}$. If $\mu'$ is invertible then we are done. If it is
not, we may iterate this procedure until invertible $\mu'$ is reached
at some lower dimension. If $\mu'$ were to remain uninvertible even on
a one-dimensional space then $\mu' = 0$, $I_{\mathcal{R}'}$ would be
identically zero, and therefore $I_{\mathcal{R}}$ would be identically
zero. By assumption, this is not the case.
\end{proof}

\end{document}